\pdfoutput=1
\documentclass[a4paper,USenglish,cleveref,autoref,thm-restate,pdfa]{lipics-v2021}
\hideLIPIcs
\nolinenumbers 
\usepackage[linesnumbered, ruled]{algorithm2e}
\title{Online and Incremental Fractional Vertex Cover\\ on Trees} 
\titlerunning{Online and Incremental Fractional Vertex Cover on Trees}
\author{Júlia Baligács}
{Department of Computer Science, University of Oxford, United Kingdom}
{jbaligacs@gmail.com}
{https://orcid.org/0000-0003-2654-149X}
{During employment in Warsaw supported by ERC grant BOBR (grant no.~948057).
During employment in Oxford supported by ERC grant CCOO (grant no.~101165139).}

\author{Bart{\l}omiej Bosek}
{Institute~of~Theoretical~Computer~Science, Faculty~of~Mathematics~and~Computer~Science, Jagiellonian~University, Krak{\'o}w, Poland}
{bartlomiej.bosek@uj.edu.pl}
{https://orcid.org/0000-0001-8756-3663}
{}

\author{Yann Disser}
{Department~of~Mathematics, TU Darmstadt, Germany}
{disser@mathematik.tu-darmstadt.de}
{https://orcid.org/0000-0002-2085-0454}
{}

\author{Andreas Emil Feldmann}
{Department of Computer Science, University of Sheffield, United Kingdom}
{feldmann.a.e+web@gmail.com}
{https://orcid.org/0000-0001-6229-5332}
{}

\author{Grzegorz Gutowski}
{Institute~of~Theoretical~Computer~Science, Faculty~of~Mathematics~and~Computer~Science, Jagiellonian~University, Krak{\'o}w, Poland}
{grzegorz.gutowski@uj.edu.pl}
{https://orcid.org/0000-0003-3313-1237}
{}

\author{Katarzyna Kępińska}
{Institute~of~Theoretical~Computer~Science, Faculty~of~Mathematics~and~Computer~Science, Jagiellonian~University, Krak{\'o}w, Poland}
{katarzyna.kepinska@alumni.uj.edu.pl}
{}
{}

\author{Paweł Putra}
{Institute of Informatics, University of Warsaw, Poland}
{p.putra@uw.edu.pl}
{https://orcid.org/0009-0006-3703-680X}
{}

\author{Anna Zych-Pawlewicz}
{Institute of Informatics, University of Warsaw, Poland}
{anka@mimuw.edu.pl}
{https://orcid.org/0000-0002-5361-8969}
{}

\authorrunning{Baligács, Bosek, Disser, Feldmann, Gutowski, Kępińska, Putra, Zych-Pawlewicz}
\Copyright{Júlia Baligács, Bart{\l}omiej Bosek, Yann Disser, Andreas Emil Feldmann, Grzegorz Gutowski, Katarzyna Kępińska, Pawe{\l} Putra, Anna Zych-Pawlewicz}


\funding{This\; work\; is\;  a\; result\; of\; research\; conducted\; within\;\; research\; project\; number\; 2023/51/B/ST6/02833 financed by the National Science Centre, Poland (Bartłomiej Bosek, Grzegorz Gutowski, Katarzyna Kępińska, Paweł Putra and Anna Zych-Pawlewicz)}
\acknowledgements{}

\usepackage{tikz}
\usetikzlibrary{arrows,math,patterns,shapes,fit,decorations.pathreplacing,shapes.geometric}
\usepackage{nicefrac}

\renewcommand{\epsilon}{\varepsilon}

\let\leq\leqslant
\let\geq\geqslant

\newcommand{\cred}{\texttt{red}}
\newcommand{\cblue}{\texttt{blue}}
\newcommand{\cgreen}{\texttt{green}}
\newcommand{\dr}{\Delta_{\cred}}
\newcommand{\dg}{\Delta_{\cgreen}}
\newcommand{\db}{\Delta_{\cblue}}

\newcommand{\alg}{\textsc{Alg}\xspace}
\newcommand{\opt}{\textsc{Opt}\xspace}
\newcommand{\compratio}{\rho}
\SetKwFor{AtTime}{at time}{do}{end}

\EventEditors{Philip Bille, Seth Pettie, and Sabine Storandt}
\EventNoEds{3}
\EventLongTitle{34th Annual European Symposium on Algorithms (ESA 2026)}
\EventShortTitle{ESA 2026}
\EventAcronym{ESA}
\EventYear{2026}
\EventDate{August 31--September 4, 2026}
\EventLocation{L'Aquila, Italy}
\EventLogo{}
\SeriesVolume{388}
\ArticleNo{158}

\begin{document}

\maketitle
\enlargethispage{\baselineskip}

\begin{abstract}
In this paper we study the fractional vertex cover problem on trees in two related models: online and incremental.
In the online model, the vertices of the tree are known a priori and the edges arrive one at a time. 
The goal is to maintain a fractional vertex cover of the tree, i.e., an assignment of fractional weights from $[0,1]$ to the vertices such that the weights of endpoints of every edge sum up to at least one. 
After each edge arrival, we need to modify the fractional vertex cover to cover the new edge as well. However, we can only increase the values assigned to vertices.
The problem was studied before (in the vertex arrival model) by Wang and Wong, who motivated it as a generalization of the ski-rental problem, but also (more importantly) by its close connection to the dual online matching problem.
They presented a $1.901$-competitive algorithm for general graphs in the vertex arrival model. 
We present an $\frac{11}{6} \approx 1.83$-competitive algorithm for trees in the more general edge arrival model. 

In addition, we study the fractional vertex cover problem in an incremental model, where we again seek a fractional vertex cover after every update, but all the updates to the tree are known to the algorithm a priori. 
In this model, we give a $1.5$-competitive algorithm and provide a matching lower bound.
\end{abstract}

\section{Introduction and Related Work}
In this paper we consider the fractional vertex cover problem in the online and in the incremental edge arrival model.
In this model, edges of an underlying graph appear one by one and we need to maintain a fractional vertex cover after each step.
Here, a fractional (resp.~integral) vertex cover is an assignment of weights from $[0,1]$ (resp.~from $\{0,1\}$) to the vertices so that every edge in the graph is \emph{covered}, i.e., such that the weights assigned to its endpoints sum to at least~$1$.
The objective is to minimize the total weight of the vertices, and the challenge is that our decisions are irrevocable in the sense that we may only ever increase the weights assigned to the vertices from one step to the next.
We cannot hope to maintain an optimal fractional vertex cover in each time step, instead we aim for solutions of small \emph{competitive ratio}, i.e., solutions that are never suboptimal by more than this ratio.

The online setting differs from the incremental setting in that it assumes no prior knowledge of the edges of the graph, while the incremental setting assumes that we know the graph a priori.
Both settings have in common that we need to define increasing sequences of vertex weights that simultaneously approximate a minimum fractional vertex cover across all time steps.
The incremental setting can be viewed as an online problem, where the only unknown is the number of steps that we need to optimize for (much like in the ski rental problem).
We provide online algorithms for both the online and the incremental setting in the edge arrival model for the case where the underlying graph is a tree.

The online fractional vertex cover problem was studied on general graphs by Wang and Wong~\cite{WangW2015}, but in the vertex arrival model, where vertices are revealed one at a time together with the adjacent edges to the previously revealed vertices.
This model is easier than the edge arrival model in the sense that vertex addition can be simulated by adding adjacent edges one by one.
Wang and Wong motivated the problem as a generalization of the ski-rental problem, and, more importantly, by its connection to the fractional online matching problem.
They present a $1.901$-competitive online algorithm maintaining a fractional vertex cover, with a proof based on primal-dual analysis. 
This means that the authors maintain an online solution to the dual fractional matching problem.
By weak duality, the maintained matching is $(1/1.901 \approx 0.526)$-competitive.
This lower bound was later matched by an algorithm by Tang and Zhang~\cite{TangZhang24},
which was later shown to be optimal by Tang~\cite{Tang2026}.

The existence of a $(1/2+\epsilon)$-competitive algorithm for the fractional matching problem on general graphs in vertex arrival model had been a long standing open question.
Thus, as a byproduct of their approach, Wang and Wong provide such an online algorithm. 
Building on this result, Gamlath et~al.~\cite{Gamlath} used a general rounding technique to transform an online fractional matching algorithm to an online randomized (integral) matching algorithm. 
They obtain a randomized $(\frac{1}{2}+\epsilon)$-competitive online matching algorithm for some small~$\epsilon > 0$. 
On the other hand, Gamlath et~al.~also show that, in the edge arrival model, no online algorithm (even randomized) can beat a competitive ratio of~$1/2$. 
This is an important result, since online matching is a very well-motivated and widely studied problem. 
A straightforward greedy algorithm is $1/2$-competitive for this problem, and a significant amount of research has been devoted to providing better ratios in different regimes and for different graph classes~\cite{KarpVV90,MahdianYan11,NaorWajc18,CohenWajc18,BuchbinderJainNaor07,LeeSingla20,BuchbinderSegevTkach19}.   

For forests, the online fractional matching problem was studied by Buchbinder et~al.~\cite{BuchbinderSegevTkach19}. 
They show an online $5/9$-competitive fractional matching algorithm which translates to a randomized integral algorithm with the same expected competitive ratio. 
Their analysis relies again on the primal-dual technique, where instead of maintaining a proper fractional vertex cover throughout the runtime of the algorithm, a new one is constructed retrospectively after each edge introduction.
The construction of those covers heavily depends on the parent--child relation between edge endpoints imposed by rooting each tree in the forest arbitrarily at each time step. 
Sadly this does not yield a $(9/5=1.8)$-competitive online fractional vertex cover, 
as merging two trees in the forest by introducing an edge between them and rooting this tree can effectively reverse the relationship between some edge endpoints.
This in turn leads to a vertex weight decreasing in the cover at some steps, which is not allowed in an online setting.
They also show an upper bound of $\approx0.5914$ on the competitive ratio of any fractional matching algorithms in this model. 
Online fractional matchings on forests were also studied by Jiang and Zhang~\cite{JiangZhang2026} in the free disposal model, in which any of the matched edges can later be removed permanently by the algorithm. 
They show a $(5/8 = 0.625)$-competitive algorithm for online fractional matching in this setting, showing a clear separation between competitive ratios of algorithms attainable for the two models.
Importantly the free disposal model allows them to maintain a single dual solution throughout the algorithm
by taking advantage of the model allowing non-incremental changes, which is not possible in our case.

To the best of our knowledge, the best known competitive ratio for fractional vertex cover on trees in the edge arrival model is the trivial ratio of $2$, as the bound of Wang and Wong~\cite{WangW2015} does not apply to the edge arrival model. 
In this paper, we improve the state of the art by presenting a $(\frac{11}{6} \approx 1.83)$-competitive online fractional vertex cover algorithm for trees.

We also present a $1.5$-competitive incremental fractional vertex cover algorithm for trees, as well as a matching lower bound for the problem.
The fractional vertex cover problem does not seem to have been previously considered from an incremental perspective.
Incremental minimization problems studied in the past include $k$-median~\cite{ChrobakKenyonNogaYoung08,MettuPlaxton03,LinNagarajanRajaramanWilliamson10}, facility location~\cite{Plaxton06,LinNagarajanRajaramanWilliamson10}, and $k$-center~\cite{Gonzalez85}.
In all of these problems, solutions are only constrained in their (growing) cardinality, otherwise every solution is feasible and the structure of the problem is encoded in the objective. 
In contrast, in the incremental vertex cover problem, we have a covering constraint in the sense that only solutions covering all the edges are permitted.
As noted by Lin et al.~\cite{LinNagarajanRajaramanWilliamson10}, the results for the minimum latency problem in~\cite{BlumChalasaniCoppersmithEtAl94,GoemansKleinberg98} implicitly yield results for the incremental $k$-MST problem, which is another example of an incremental minimization problem with a non-trivial feasibility constraint, namely that the solution must form a tree.

There is a substantial body of work on incremental \emph{maximization} problems, ranging from problems under a growing cardinality constraint~\cite{BernsteinDisserGrossHimburg22,DisserKlimmSchewiorWeckbecker23,DisserWeckbecker25} to growing knapsack~\cite{MegowMestre13,DisserKlimmMegowStiller17,DisserKlimmLutzWeckbecker24} and matroid~\cite{FujitaKobayashiMakino13,KakimuraMakino13} constraints.
Hassin and Rubinstein~\cite{HassinRubinstein02} considered the incremental maximum matching problem, which is related to our setting via the duality between matchings and vertex covers.
They show that the best possible competitive ratio for incremental matchings is~$\sqrt{2}$, which was later complemented by Matuschke, Skutella, and Soto~\cite{MatuschkeSkutellaSoto18} with a tight bound of $\ln(4) \approx 1.38$ for the randomized competitive ratio.
The setting was further generalized to matroid intersection~\cite{FujitaKobayashiMakino13} and independence systems with bounded exchangeability~\cite{KakimuraMakino13}.
That being said, these results do not seem transferable to (fractional) vertex cover. 

\section{Online Fractional Vertex Cover Algorithm}
In this section, we present our online fractional vertex cover algorithm for trees that we call RGB algorithm. Before we proceed, we more formally introduce the necessary definitions.

A \emph{fractional vertex cover}
of a graph $G=(V,E)$ is a function $f: V \to [0, 1]$ such
that, for every edge $(u,v)\in E$, we have $f(u) + f(v) \geq 1$.
The online edge arrival model assumes that
the edges of $G$ are revealed one at a time according to some \emph{presentation order} $e_1 \ll e_2 \ll \cdots \ll e_m$.
We denote by $G_t:=(V,\{e_1, \dots, e_t\})$ the graph after presenting all edges $e_i$ with $i \leq t$.
The task of an online algorithm is to produce a sequence $f_1, \ldots, f_m$ of fractional vertex covers such that:

\begin{enumerate}[(1)]
  \item For all $t \in \{1, \ldots, m\}$, $f_t$ is a fractional vertex cover for $G_t$.
  \item For all $t \in \{2, \ldots, m\}$ and $v\in V$, we have $f_{t-1}(v) \leq f_{t}(v)$, i.e., the covers are increasing.
\end{enumerate}

Importantly, the algorithm needs to output $f_t$ without any knowledge of $G_{t+1}, \dots, G_m$.
By $\opt_t$, we denote the weight of a minimum vertex cover in $G_t$.
For a fixed algorithm $\alg$, by $\alg_t := \sum_{v\in V} f_t(v)$, we denote the total weight of the fractional vertex cover generated by it at time $t$.
We say that the algorithm $\alg$ is~$\compratio$-competitive if $\max\{\frac{\alg_t}{\opt_t} \mid t \in [m]\} \leq \compratio$ for any input.
As the algorithm has no knowledge of the future, we can equivalently say that an algorithm is $\compratio$-competitive if $\alg_m\leq \compratio \cdot  \opt_m$ for every instance on $m$ edges.
When $t$ is clear from context we omit the subscript and also use $y_v$ as a shorthand of~$f_t(v)$.
If for an edge $(u, v)\in E$ it holds that $y_u + y_v \geq 1$, we say that $(u,v)$ is \emph{covered}.
If, on the other hand, $y_u + y_v < 1$, we say that $(u,v)$ is \emph{uncovered} with a \emph{deficit} of $d = 1 - y_u - y_v$.

A \emph{fractional matching} for a graph $G$ is an assignment of weights $m_{uv} \in [0, 1]$ to each edge~$(u, v) \in E$ such that for every $v \in V$ it holds that $\sum_{(u,v) \in E} m_{uv} \leq 1$, i.e., the sum of weights on edges incident to
any vertex does not exceed $1$. In the fractional matching problem, the objective is to maximize the sum of the values $m_{uv}$ assigned to the edges.

It is worth noting that on bipartite graphs (and in particular on trees) both the fractional vertex cover problem and the fractional matching problem admit optimal integral solutions.
In particular, $\opt_t$ will always be integral for our purposes.
We also note that it is very easy to obtain a competitive ratio of $2$ for online fractional vertex cover problem: we maintain a maximal matching and pick the endpoints of the matching edges to the vertex cover (i.e., assign them the values $1$).
Clearly, the size of a minimum vertex cover is at least the size of the matching, and hence $2$-competitiveness follows.

We now move on to defining the RGB algorithm.

\begin{definition}[RGB algorithm]\label{def:rgb}
The RGB algorithm colors each arriving edge of the graph by one of three available colors from the list $[\cred=1, \cgreen=2, \cblue=3]$.
On arrival of an edge $e$, it is assigned color $c_e$, which is either the first color from the list not present on its incident edges,
or $\cblue$ if no such color exists (see Algorithm~\ref{alg:rgb}).
Let $C(v)$ be a set of colors present on the edges adjacent to $v$. The weight of a vertex in the cover is assigned according
to the colors of adjacent edges in the following way
\begin{equation*}
f(v)=
\begin{cases}
1 \text{ if } C(v)=\{ \cred, \cgreen, \cblue \}\\
\frac{5}{6} \text{ if } C(v)=\{ \cred, \cgreen \}\\
\frac{4}{6} \text{ if } C(v)=\{ \cred, \cblue \}\\
\frac{2}{6} \text{ if } C(v)=\{ \cgreen, \cblue \}\\
\frac{3}{6} \text{ if } C(v)=\{ \cred \}\\
\frac{1}{6} \text{ if } C(v)=\{ \cgreen \}\\
0 \text{ if } C(v)=\{ \cblue \} \text{ or } C(v)=\emptyset.
\end{cases}
\end{equation*}
\end{definition}

\begin{algorithm}
\caption{RGB Algorithm}
  \label{alg:rgb}
  \KwIn{Online graph \( G = (V, E) \) with offline vertices \( V \)}
  \KwOut{A fractional vertex cover of \( G \)}
  Initialize \( y_u \gets 0 \) for each \( u \in V \)\\
  \For{each online edge \( e=(u, v) \)}{
    $c_e \gets \min{\{\{ i \in \{1, 2\} \mid u,v \text{ are not adjacent to any edge of color } i \} \cup \{3\}\}}$\\
    $y_u \gets f(u)$\\
    $y_v \gets f(v)$
  }

  \Return $(y_v)_{v \in V}$
\end{algorithm}

We start the analysis by proving the correctness of RGB algorithm.

\begin{lemma}
  The RGB Algorithm maintains a fractional vertex cover after every edge arrival.
\end{lemma}
\begin{proof}
  We need to prove that after introducing an online edge $e = (u,v)$ it holds that $y_u + y_v \geq 1$. We make case distinction for each possible color assigned to the new edge $e$:

\begin{itemize}
  \item $c_e = \cred$: Both $u$ and $v$ are incident to a $\cred$ edge, implying that $y_u, y_v \geq \frac{1}{2} $, so $ y_u + y_v \geq 1$.
    
  \item $c_e = \cgreen$: Since $e = (u, v)$ got assigned \cgreen, it means that at least one of $u, v$ must be adjacent to a $\cred$ edge. Assume without loss of generality that it is $u$, then $\{\cred, \cgreen\} \subseteq C(u)$ giving $y_u \geq \frac{5}{6}$, and $\cgreen \in C(v)$ giving $y_v \geq \frac{1}{6}$, so $ y_u + y_v \geq 1$
    
  \item $c_e = \cblue$: This means that at least one of $u, v$ is adjacent to a $\cred$ edge and at least one is adjacent to a $\cgreen$ edge.
    There are two cases:
    \begin{itemize}
      \item Case 1: At least one of $u, v$ is adjacent to both a $\cred$ and a $\cgreen$ edge. Assume without loss of generality that $u$ is such a vertex. Then $\{\cred, \cgreen, \cblue\} \subseteq C(u)$ giving $y_u \geq 1$, so $ y_u + y_v \geq 1$.
      \item Case 2: Both $u$ and  $v$ are adjacent to exactly one of red and green. Assume without loss of generality that $u$ is adjacent to $\cred$ and $v$ is adjacent to $\cgreen$. Then we have $\{\cred, \cblue\} \subseteq C(u)$, giving $y_u \geq \frac{4}{6}$ and $\{\cgreen, \cblue\} \subseteq C(v)$ giving $y_v \geq \frac{2}{6}$, so $ y_u + y_v \geq 1$
    \end{itemize}

\end{itemize}

\end{proof}

Next, we prove that the RGB algorithm achieves the desired competitive ratio.

\begin{theorem}
  The RGB Algorithm is $\frac{11}{6}$-competitive on forests.
\end{theorem}

\begin{proof}
  For a fixed moment $t$ we will show how to transform the total weight of the cover $\alg_t$ produced up to this moment into a fractional matching $M$ of weight $\frac{6}{11}\alg_t$.
  The result will follow from weak duality, as the vertex cover LP is the dual of the matching LP (see for instance~\cite{DevanurJainKleinberg13}).
  
  Let $m_{uv}$ denote the weight discharged onto the edge $(u,v)$.
  To create the matching we first root each connected component at an arbitrary vertex.
  Then for each vertex $v$ we discharge the weight $y_v$ by doing the following.

  \begin{enumerate}
    \item If vertex $v$ has a red edge $(u_r, v)$, then increase $m_{u_r v}$ by $\dr$ and decrease $y_v$ by $\dr$, where $\dr = \min(y_v, \frac{4}{6} - \frac{1}{6}[u_r \text{ is a parent of } v])$.
      
    \item If vertex $v$ has a green edge $(u_g, v)$, then increase $m_{u_g v}$ by $\dg$ and decrease $y_v$ by $\dg$, where $\dg = \min(y_v, \frac{2}{6} - \frac{1}{6}[u_g \text{ is a parent of } v])$.

    \item If vertex $v$ has at least one blue edge, fix one such edge $(u_b, v)$, increase $m_{u_b v}$ by $\db$ and decrease $y_v$ by $\db$, where $\db = \min(y_v, \frac{1}{6} - \frac{1}{6}[u_b \text{ is a parent of } v])$.
  \end{enumerate}

  To prove that this is enough to discharge the entirety of $y_v$, it suffices to show that the sum of discharged amounts is at least $y_v$ for all possible sets $C(v)$ of colors incident to $v$. 
  For $C(v) = \{\cred, \cgreen, \cblue\}$, by summing the discharged amounts over all steps, we obtain that the discharged amount is at least
  \begin{equation*}
     \min\left(y_v, \frac{4}{6} + \frac{2}{6} + \frac{1}{6} - \frac{1}{6}[\text{some } u \in \{u_r, u_g, u_b\} \text{ is a parent of } v] \right)=
     \min\left(y_v, \frac{6}{6}\right) = y_v,
  \end{equation*}
  as $y_v \leq 1$ and each vertex $v$ has at most one parent.
  For $C(v)=\{ \cred, \cgreen \}$ we get
  \begin{equation*}
     \min\left(y_v, \frac{4}{6} + \frac{2}{6} - \frac{1}{6}[\text{some } u \in \{u_r, u_g\} \text{ is a parent of } v] \right)=
     \min\left(y_v, \frac{5}{6}\right) = y_v,
  \end{equation*}
  as $y_v \leq \frac{5}{6}$. For $C(v)=\{ \cred, \cblue \}$ we get
  \begin{equation*}
     \min\left(y_v, \frac{4}{6} + \frac{1}{6} - \frac{1}{6}[\text{some } u \in \{u_r, u_b\} \text{ is a parent of } v] \right)=
     \min\left(y_v, \frac{4}{6}\right) = y_v,
  \end{equation*}
  as $y_v \leq \frac{4}{6}$. For $C(v)=\{ \cgreen, \cblue \}$ we get
  \begin{equation*}
     \min\left(y_v, \frac{2}{6} + \frac{1}{6} - \frac{1}{6}[\text{some } u \in \{u_g, u_b\} \text{ is a parent of } v] \right)=
     \min\left(y_v, \frac{2}{6}\right) = y_v,
  \end{equation*}
  as $y_v \leq \frac{2}{6}$.   For $C(v)=\{ \cred \}$ we get
  \begin{equation*}
     \min\left(y_v, \frac{4}{6} - \frac{1}{6}[u_r \text{ is a parent of } v] \right)=
     \min\left(y_v, \frac{3}{6}\right) = y_v,
  \end{equation*}
  as $y_v \leq \frac{3}{6}$. For $C(v)=\{ \cgreen \}$ we get
  \begin{equation*}
     \min\left(y_v, \frac{2}{6} - \frac{1}{6}[ u_g \text{ is a parent of } v] \right)=
     \min\left(y_v, \frac{1}{6}\right) = y_v,
  \end{equation*}
  as $y_v \leq \frac{1}{6}$. For $C(v)=\{ \cblue \}$ we have
$y_v =0$ so there is nothing to discharge.
  Let $m_v$ denote the sum of weights of edges incident to vertex $v$.
  Clearly $v$ can be adjacent to at most one $\cred$ and at most one $\cgreen$ edge.
Additionally, a vertex may be adjacent to multiple $\cblue$ edges, but only one of them can have a positive weight. 
To prove this, consider two cases:
    \begin{itemize} 
      \item (1) $v$ has a $\cblue$ edge to its parent $p$: In this case both $\cred$ and $\cgreen$ edges of $v$ point to its children, so the entirety of $y_v$ will be discharged onto them. This means that the only $\cblue$ edge with positive charge can be the edge $(v, p)$, as $\cblue$ edges can only get charge from their parent vertex. 

      \item (2) $v$ has no $\cblue$ edge to its parent: Since $\cblue$ edges can only get charge from their parent vertex, and $v$ charges only one $\cblue$ edge, there is only one $\cblue$ edge with positive charge incident to $v$. 
    \end{itemize}

  Summing the contributions of edges around $v$ we get
$$
    m_v \leq \underbrace{\frac{4}{6} + \frac{3}{6}}_{\text{$\cred$}} 
           + \underbrace{\frac{2}{6} + \frac{1}{6}}_{\text{$\cgreen$}} 
           + \max\Bigl(
            \underbrace{\frac{1}{6} + 0}_{\text{$\cblue$ (1)}},
            \underbrace{0 + \frac{1}{6}}_{\text{$\cblue$ (2)}}
            \Bigr)
           = \frac{7}{6} + \frac{3}{6} + \frac{1}{6} = \frac{11}{6},
$$ 
  where under-braced parts show the contributions of each edge as a sum of contribution of an edge of given color from its parent and child end in the tree.
  This gives us that $M = \left\{\frac{6}{11}m_{uv} | (u,v) \in E\right\}$ is a fractional matching.

  Let $|M^*|$ be the weight of a maximum fractional matching. From weak duality we get
  $$ \frac{6}{11}\alg = |M| \leq |M^*| = \opt$$
  $$ \alg \leq \frac{11}{6}\opt $$

\end{proof}

\section{Incremental Vertex Cover}

\newcommand{\twochild}{\ensuremath{\textsc{TwoChildRule}}\xspace}

In this section, we study a variant of the online vertex cover problem in which the algorithm has full knowledge of the future, and the objective is to maintain a good solution at all times.
The motivation for this model is the following. Consider the process of building a road network. During construction, the site must be monitored by security cameras. Given the full construction plan, we must decide where and when to install cameras. Once installed, they cannot be moved, and at every point in time we aim to minimize the total number of cameras in use.

We formalize this setting as the \emph{incremental vertex cover} problem.
Similar to the online setting, 
we are given as input a graph~$G=(V,E)$ together with an ordering of its edges $e_1 \ll \dots \ll e_m$. The goal is to compute an increasing sequence $f_1, \ldots, f_m$ of fractional vertex covers such that:

\begin{enumerate}[(1)]
  \item For all $t \in \{1, \ldots, m\}$, $f_t$ is a fractional vertex cover for $G_t=(V,\{e_1, \dots, e_t\})$.\label{cond:frac_vertex_cover}
  \item For all $t \in \{2, \ldots, m\}$ and $v\in V$, we have $f_{t-1}(v) \leq f_{t}(v)$, i.e., the covers are increasing.\label{cond:increasing}
\end{enumerate}

We denote $f_t(v)$ as $y_v^{(t)}$.
The key difference to the online setting is that the algorithm has full knowledge of $G$ and the ordering of the edges $e_1 \ll \dots \ll e_m$.
The challenge here is to compute a good solution for all times.
As in the online setting, we are seeking algorithms with a small competitive ratio, which is defined by 
$
\max_{t \in [m]} \frac{\alg_t}{\opt_t}.
$

In this section, we focus on the fractional incremental vertex cover problem on forests. We show that the optimal competitive ratio is exactly $3/2$. Specifically, we present an algorithm achieving this ratio and a matching lower bound. The algorithm is half-integral and follows a simple local rule.
We conjecture that the same ratio of $3/2$ is achievable in the integral setting.
However, we construct an example showing that no purely local integral rule can achieve a competitive ratio of $3/2$, illustrating that the integral case is substantially more challenging.

Our lower bound construction is based on a simple instance -- a binary tree revealed layer by layer. Despite its simplicity, the analysis is technically involved, as it must capture the behavior of an arbitrary fractional algorithm. The key idea is to derive constraints on the average weight assigned to each layer over time and to use LP duality to show that no algorithm satisfying these constraints can achieve a competitive ratio better than $3/2$.

\subsection{A 1.5-competitive half-integral algorithm}

In this subsection, we introduce a half-integral algorithm for incremental vertex cover, called $\twochild$, and prove that it achieves a competitive ratio of $3/2$.

Before describing the algorithm, let us first establish a useful property.
Observe that larger instances are more difficult to solve in the following sense: 
For a forest $F$, let $\compratio(F)$ denote the best achievable competitive ratio over all instances defined on $F$. If $F$ is a subgraph of $F'$, then $\compratio(F) \leq \compratio(F')$: Indeed, any instance on $F$ can be embedded into $F'$ by revealing all edges in $E(F)$ first and postponing all edges in $E(F') \setminus E(F)$ to the end.
As a consequence, we can assume wlog.~that $F$ is a tree (instead of only a forest),
for which we can pick a root, and we can assume that every vertex has either at least two children or is a leaf.

The algorithm $\twochild$ is formally described in Algorithm~\ref{alg:leftfirst}. Throughout the description, we use the convention that whenever $y_v^{(t)}$ is not explicitly specified, it remains unchanged, i.e., $y_v^{(t)} = y_v^{(t-1)}$.
The algorithm exploits the parent--child structure of the tree and follows a simple local rule. Vertices with two children are eventually assigned weight~$1$. When an edge $e=(p,c)$ arrives, where $c$ is the first child of $p$, the algorithm looks ahead to the next time step at which either $p$ or $c$ receives an additional child. If $p$ receives another child before $c$ does, then $p$ is assigned weight $1$. Otherwise, both $p$ and $c$ are assigned weight~$1/2$.

\begin{algorithm}
\setcounter{AlgoLine}{0}
\caption{$\twochild$}
\label{alg:leftfirst}
root the final tree arbitrarily\;
initialize $y_v^{(0)}=0$ ($v\in V$)\;
\AtTime{$t = 1,\ldots,m$}{
    let $e_t=(p,c)$ where $p$ is the parent of $c$\;
	\If{$y_p^{(t-1)}+y_c^{(t-1)}<1$}{
    \If{$p$ has two children or ($p$ gets another child before $c$ and $y_c^{(t-1)}=0)$}{
        $y_p^{(t)} \gets 1$\;
    }
    \Else{
        $y_p^{(t)} \gets \frac{1}{2}$,
        $y_c^{(t)} \gets \frac{1}{2}$
    }
}}
\end{algorithm}

\begin{observation}
\twochild is correct, that is, $(y_v^{(t)})_{v\in V, t \in [m]}$ satisfies conditions (\ref{cond:frac_vertex_cover}) and (\ref{cond:increasing}).
\end{observation}

\begin{proof}
By construction, the weights $(y_v^{(t)})$ cover the edge $e_t$ at every time step $t$.
Therefore, it only remains to check that no weight is ever decreased.
Since no weight is set to $0$, it suffices to ensure that no weight is changed from $1$ to $\frac12$.
  However, once a vertex has weight~$1$, all future edges incident to it are already covered, and hence no further updates involving this vertex are performed.
\end{proof}

We now proceed with proving that \twochild achieves a competitive ratio of $3/2$.
As a first step, we establish the following simple properties.

\begin{observation}\label{obs:upper_bound}
Fix a time step $t$ and consider the graph $G_t=(V, \{e_1, \dots, e_t\})$ with the vertex weights $(y_v^{(t)})_{v\in V}$ defined by $\twochild$.
\begin{enumerate}[(a)]
  \item If $u$ has no children and $y_u^{(t)} \neq 0$, then $y_u^{(t)}=\frac12$, $u$ has no siblings and $y_p^{(t)}=\frac12$ as well for the parent $p$ of $u$.\label{obs:leaf}
    \item If a vertex $u$ has exactly one child $c$ and $y_u^{(t)}=1$, then $y_c^{(t)}=0$.\label{obs:one_child}
\end{enumerate}
\end{observation}

\begin{proof}
For part~(\ref{obs:leaf}), note that the only way a leaf $u$ can receive a non-zero weight is upon arrival of edge $(u,p)$ and the second case of the algorithm is triggered.
Thus, at the time $t'$ when $(u,p)$ is revealed, we set $y_u^{(t')}=y_p^{(t')}=\frac12$.
Moreover, it holds at time $t'$ that $u$ is the only child of $p$ and $p$ does not get another child before $u$ gets a child.
Since $u$ is a leaf at time~$t$, it follows that, between times $t'+1$ and $t$, neither $u$ nor $p$ receives a new child.
Therefore, neither $y_u$ nor $y_p$ is updated to $1$ before time $t$, and hence $y_u^{(t)}=y_p^{(t)}=\frac12$.
Moreover, it still holds at time $t$ that $u$ is the only child of $p$.

  For part (\ref{obs:one_child}), if $y_u^{(t)}=1$, then this weight must have been assigned through execution of the first case upon arrival of edge~$(u,c)$ at time $t'$.
  Since $u$ has only one child at time $t'\leq t$, the fact that the if-condition in line~6 was triggered implies that~$u$ receives another child before~$c$ and $y_c^{(t')}=y_c^{(t'-1)}=0$.
Since~$u$ still has only one child at time $t$, no new edge incident to $c$ appears between times~$t'+1$ and~$t$.
Therefore, $y_c^{(t)}=0$.
\end{proof}

Now, we have all the prerequisites at hand to prove an upper bound on the competitive ratio of \twochild.

\begin{theorem}
$\twochild$ is $3/2$-competitive.
\end{theorem}

\begin{proof}
Fix a time step $t$ and let $G_t$ be the forest consisting of the edges that appeared up to time $t$.
Let $C^*$ be a minimum vertex cover of $G_t$, so that $|C^*|=\opt_t$.
Recall that
$
\alg_t = \sum_{v\in V} y_v^{(t)}
$
and we have to prove that $\alg_t \leq \frac{3}{2} \, \opt_t$.

Our strategy for this is as follows. We equip every $v\in C^*$ with a budget of $3/2$ and redistribute this budget among $v$ and its neighbors. We then show that every $u\in V$ receives at least a budget of $y_u^{(t)}$ from this redistribution. Note that this implies the theorem.

We now describe the redistribution. For each $v\in C^*$, the process is independent of the other vertices. In the following description, any budget that is not explicitly distributed is kept by $v$.
\begin{enumerate}
    \item If $y_v^{(t)}=1$, then $v$ gives $\frac12$ to its parent.
    \item If $y_v^{(t)}=\frac12$, then $v$ gives $\frac12$ to its parent and $\frac12$ to its first child.
    \item If $y_v^{(t)}=0$, then $v$ gives $1$ to its parent and $\frac12$ to its first child.
\end{enumerate}

It remains to show that every $u\in V$ receives at least a budget of $y_u^{(t)}$.
First, if $u\in C^*$, then $u$ keeps at least a budget of $y_u^{(t)}$ to itself in all three cases.
Therefore, assume from now on that $u\notin C^*$. In particular, all neighbors of $u$ belong to $C^*$.
We distinguish three cases: (1) where $u$ has at least two children, (2) where $u$ has no children, and (3) where $u$ has precisely one child.

For case (1), note that, in all three redistribution rules, every vertex gives at least $\frac12$ to its parent.
Since $u$ has at least two children and all of them belong to $C^*$, $u$ receives in total at least a budget of 1 from the redistribution.

Next, we consider case (2), i.e., where $u$ has no children.
If $y_u^{(t)}=0$, there is nothing to prove. Otherwise, Observation~\ref{obs:upper_bound}\ref{obs:leaf} implies that $u$ has no siblings and
$
y_u^{(t)}=y_p^{(t)}=\frac12,
$
where~$p$ is the parent of $u$. Since $p\in C^*$, it gives $\frac12$ to $u$, so $u$ is covered.

Last, we consider case (3), i.e., where $u$ has precisely one child $c$.
If $y_u^{(t)}=1$, then by Observation~\ref{obs:upper_bound}\ref{obs:one_child}, we have $y_c^{(t)}=0$. Since $c\in C^*$, it gives $1$ to its parent $u$, and thus~$u$ is covered.
If $y_u^{(t)}=\frac12$, then $c\in C^*$ gives at least $\frac12$ to its parent in all cases, so again $u$ is covered.

This exhausts all possibilities and therefore completes the proof of the theorem.
\end{proof}

It remains open what the competitive ratio of the integral version of the problem is.
For this, we formulate the following conjecture, which we verified for all instances on at most 10 vertices via exhaustive search with a computer program.

\begin{conjecture}
There exists an integral $3/2$-competitive algorithm for incremental vertex cover on forests.
\end{conjecture}

In the following, we give an example showing that a simple local rule similar to the one used for the half-integral algorithm cannot achieve the claimed competitive ratio of $3/2$.
This suggests that solving the integral version of the problem is substantially more difficult.

\begin{proposition}
\label{prop:integral_lower}
Consider an integral algorithm for incremental vertex cover on trees.
Assume that the tree is rooted and, upon arrival of edge $e=(u,v)$, the decision of the algorithm depends only on the subtree consisting of ancestors and descendants of $u$ and $v$ (and the ordering of its edges).
Then the competitive ratio of the algorithm is at least~$\frac53$.
\end{proposition}

\begin{figure}
\begin{center}
\begin{tikzpicture}[scale=0.95, every node/.style={circle, draw, inner sep=0pt, minimum size=10pt}]
\node (a) at (0,0) {};
\node (b) at (-1,-1) {};
\node (c) at (1,-1) {};
\node (d) at (-1,-2) {};
\node (e) at (1,-2) {};
\node (f) at (-1,-3) {};
\node (g) at (1,-3) {};
\node (h) at (-1,-4) {};
\node (i) at (1,-4) {};
\node (j) at (-2,-2) {};
\node (k) at (2,-2) {};
\draw (a) to node [left, draw=none,  yshift=3pt] {1} (b);
\draw (a) to node [right, draw=none,  yshift=3pt] {2} (c);
\draw (b) to node [left, draw=none] {3} (d);
\draw (c) to node [right, draw=none] {4} (e);
\draw (d) to node [left, draw=none] {5} (f);
\draw (e) to node [right, draw=none] {6} (g);
\draw (f) to node [left, draw=none] {7} (h);
\draw (g) to node [right, draw=none] {8} (i);
\draw (b) to node [left, draw=none, yshift=3pt] {9} (j);
\draw (c) to node [right, draw=none,  yshift=3pt] {10} (k);
\end{tikzpicture}
\hspace{1mm}
\begin{tikzpicture}[scale=0.95, every node/.style={circle, draw, inner sep=0pt, minimum size=10pt}]
\node (a) at (0,0) [fill, rectangle] {};
\node (b) at (-1,-1) [fill] {};
\node (c) at (1,-1) [fill] {};
\node (d) at (-1,-2) [rectangle] {};
\node (e) at (1,-2) [rectangle] {};
\node (f) at (-1,-3) [fill] {};
\node (g) at (1,-3) [fill] {};
\phantom{
\node (h) at (-1,-4) {};
\node (i) at (1,-4) {};
\node (j) at (-2,-2) {};
\node (k) at (2,-2) {};
}
\draw (a) to node [left, draw=none,  yshift=3pt] {1} (b);
\draw (a) to node [right, draw=none,  yshift=3pt] {2} (c);
\draw (b) to node [left, draw=none] {3} (d);
\draw (c) to node [right, draw=none] {4} (e);
\draw (d) to node [left, draw=none] {5} (f);
\draw (e) to node [right, draw=none] {6} (g);
\end{tikzpicture}
\hspace{1mm}
\begin{tikzpicture}[scale=0.95, every node/.style={circle, draw, inner sep=0pt, minimum size=10pt}]
\node (a) at (0,0) [fill] {};
\node (b) at (-1,-1) [fill, rectangle] {};
\node (c) at (1,-1) [fill, rectangle] {};
\node (d) at (-1,-2) [fill] {};
\node (e) at (1,-2) [fill] {};
\node (f) at (-1,-3) [rectangle] {};
\node (g) at (1,-3) [rectangle] {};
\node (h) at (-1,-4) [fill] {};
\node (i) at (1,-4) [fill] {};
\node (j) at (-2,-2) {};
\node (k) at (2,-2) {};
\draw (a) to node [left, draw=none,  yshift=3pt] {1} (b);
\draw (a) to node [right, draw=none,  yshift=3pt] {2} (c);
\draw (b) to node [left, draw=none] {3} (d);
\draw (c) to node [right, draw=none] {4} (e);
\draw (d) to node [left, draw=none] {5} (f);
\draw (e) to node [right, draw=none] {6} (g);
\draw (f) to node [left, draw=none] {7} (h);
\draw (g) to node [right, draw=none] {8} (i);
\draw (b) to node [left, draw=none, yshift=3pt] {9} (j);
\draw (c) to node [right, draw=none,  yshift=3pt] {10} (k);
\end{tikzpicture}
\end{center}
\caption{Example of a symmetric instance of incremental vertex cover where no symmetric integral solution achieves a competitive ratio of $3/2$.
The left subfigure illustrates the instance and the other two subfigures illustrate the two possible symmetric solutions. The filled vertices are the ones chosen by the incremental algorithm, and the square vertices denote the offline optimum at times 6 and 10, respectively.}
\label{fig:integral_example}
\end{figure}
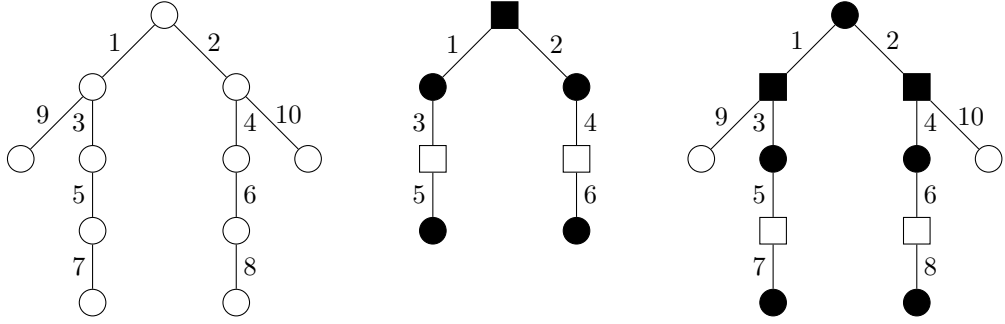

\begin{proof}
Consider the example depicted on the left in Figure~\ref{fig:integral_example}, where the root vertex $r$ is the vertex depicted at the top.
Observe that the two subtrees rooted at the two children of $r$ are equivalent in the sense that they result in the same tree and their edges appear in the same order.
However, we argue that any integral algorithm that follows the same strategy in both subtrees has competitive ratio at least $\frac53$.

At time $t=2$, the only integral solution with competitive ratio less than 2 is to select the root vertex, so we can assume that the algorithm does so.
At times 3 and 4, the algorithm needs to decide whether it selects the parent vertex or child vertex. We assume that the algorithm decides for the same in the two subtrees.
If it selects in both the parent vertex (see second subfigure), at times 5 and 6, it needs to select an additional vertex each (in the figure, we have decided wlog.~for the child vertex). Then, at time 6, we have $\alg_6=5$, while $\opt_6=3$.
Therefore, assume that, at times 3 and~4, the algorithm selects in both subtrees the child vertex (see third subfigure). At times 7 and 8, and at times 9 and 10, an additional vertex must be selected each. 
Then, at time 10, we have $\alg_{10}=7$, while $\opt_{10}=4$.
To summarize, in either case, the competitive ratio of the algorithm is at least $\frac53$.
\end{proof}

However, it is not difficult to see that an ``asymmetric'' $3/2$-competitive integral solution of the example presented in Proposition~\ref{prop:integral_lower} exists (see Figure~\ref{fig:asymmetric_solution}).

\begin{figure}
\begin{center}
\begin{tikzpicture}[scale=0.95, every node/.style={circle, draw, inner sep=0pt, minimum size=11pt}]
\node (a) at (0,0)  {1};
\node (b) at (-1,-1)  {};
\node (c) at (1,-1) {4};
\node (d) at (-1,-2) {3};
\node (e) at (1,-2) {};
\node (f) at (-1,-3)  {};
\node (g) at (1,-3) {6};
\node (h) at (-1,-4) {7};
\node (i) at (1,-4)  {};
\node (j) at (-2,-2) {9};
\node (k) at (2,-2) {};
\draw (a) to node [left, draw=none,  yshift=3pt] {1} (b);
\draw (a) to node [right, draw=none,  yshift=3pt] {2} (c);
\draw (b) to node [left, draw=none] {3} (d);
\draw (c) to node [right, draw=none] {4} (e);
\draw (d) to node [left, draw=none] {5} (f);
\draw (e) to node [right, draw=none] {6} (g);
\draw (f) to node [left, draw=none] {7} (h);
\draw (g) to node [right, draw=none] {8} (i);
\draw (b) to node [left, draw=none, yshift=3pt] {9} (j);
\draw (c) to node [right, draw=none,  yshift=3pt] {10} (k);
\end{tikzpicture}
\end{center}
\caption{An asymmetric integral $3/2$-competitive solution to the instance presented in Proposition~\ref{prop:integral_lower}. The vertex cover for time $t$ consists of the vertices with a label $\leq t$.}
\label{fig:asymmetric_solution}
\end{figure}
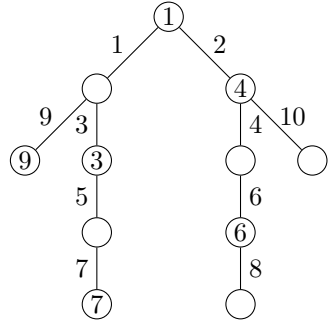

\subsection{A lower bound of 1.5 for fractional incremental vertex cover on trees}

In this subsection, we prove that the algorithm \twochild is optimal by establishing a matching lower bound.

\begin{theorem}\label{thm:lower-bound}
Every algorithm for fractional incremental vertex cover on trees has a competitive ratio of at least $3/2$.
\end{theorem}

\begin{proof}
The instance that we consider is a rooted binary tree of some depth $d$, which is revealed layer-by-layer.
More precisely, \emph{layer $i$} denotes the set of vertices at distance $i$ from the root and, in \emph{phase $i$}, we reveal all edges between layer $i-1$ and $i$. When we refer to \emph{phase 0}, we mean the time prior to arrival of any edge.

Let $\alg$ be an algorithm for the incremental fractional vertex cover problem. 
For the analysis, we only consider the time steps at the end of a phase.
For $i=1, \dots, d$, let $\opt_i$ denote the weight of an optimal offline vertex cover after phase $i$ and let $\alg_i$ denote the cost of the algorithm after phase $i$.

We begin by estimating $\opt_i$. For this, note that equipping the vertices of layers $i-1, i-3, i-5, \dots$ with weight 1 and all other vertices with weight 0 yields a valid fractional vertex cover. For $i$ odd, its total weight is at most
\begin{equation*}
\opt_i \leq 2^{i-1} + \dots + 1 = \sum\limits_{j=0}^{(i-1)/2} 4^j = \frac{4^{(i+1)/2}-1}{4-1} \leq \frac{2^{i+1}}{3}.
\end{equation*}
For $i$ even, its total weight is at most
\begin{equation*}
\opt_i \leq 2^{i-1} + \dots + 2
= 2 \cdot  \sum\limits_{j=0}^{(i-2)/2} 4^j = 2 \cdot \frac{4^{i/2}-1}{4-1} \leq \frac{2^{i+1}}{3}.
\end{equation*}
In any case, we have
\begin{equation}\label{eq:opt}
\opt_i\leq \frac{2^{i+1}}{3}.
\end{equation}

Next, we analyze the behavior of $\alg$. For this, let $\alpha_i$ be the average weight of vertices in layer $i$ after phase $i$ ($i\in \{0, \dots, d\}$), and let $\beta_i$ be the average weight of vertices in layer $i$ after phase $i+1$ ($i\in \{0, \dots, d-1\}$). Note that, after phase $i+1$, the weights of the vertices in layer $i$ are not further increased as no new adjacent edges are revealed.

We derive some inequalities that $\alg$ must fulfill. First, incrementality ensures that
\begin{equation}\label{eq:incremental}
\alpha_i \leq \beta_i \text{ for all } i\in \{0, \dots, d-1\}.
\end{equation}
Next, after phase $i+1$, the edges between layer $i$ and $i+1$ need to be covered. This implies
\begin{equation}\label{eq:covering}
\beta_i + \alpha_{i+1} \geq 1 \text{ for all } i\in \{0, \dots, d-1\}.
\end{equation}
To see this, note that, after phase $i+1$, the total weight of vertices in layer $i$ is $2^i \beta_i$ and the total weight of vertices in layer $i+1$ is $2^{i+1}\alpha_{i+1}$. Each of the vertices in layer $i$ contributes its weight to two edges between layer $i$ and $i+1$ and each of the vertices in layer $i+1$ contributes its weight to one such edge. Therefore, the edges between layer $i$ and $i+1$ get a total weight of $2\cdot 2^{i} \beta_i+2^{i+1}\alpha_{i+1}$. Since there are $2^{i+1}$ such edges, we obtain $\beta_{i}+\alpha_{i+1}\geq 1$.

Next, we have
\begin{equation}\label{eq:alg}
\alg_i=2^i \alpha_i + \sum\limits_{j=0}^{i-1}2^j \beta_j \text{ for all } i\in \{1, \dots, d\}.
\end{equation}

Let $\compratio$ denote the competitive ratio of $\alg$. We obtain
\begin{equation}\label{eq:comp}
\compratio\geq \frac{\alg_i}{\opt_i}
\overset{\eqref{eq:opt} \eqref{eq:alg}}{\geq} \frac{2^i \alpha_i + \sum\limits_{j=0}^{i-1}2^j \beta_j}{2^{i+1}/3}
= \frac{3}{2}\alpha_i + \sum\limits_{j=0}^{i-1} \frac{3}{2^{i+1-j}} \beta_j
 \text{ for all } i\in \{1, \dots, d\}.
\end{equation}

Putting inequalities \eqref{eq:incremental}, \eqref{eq:covering}, \eqref{eq:comp} together, we obtain that the competitive ratio of $\alg$ is at least the optimal objective value of the following linear program
\begin{align*}
\mathrm{(LP)}\qquad  \min \quad& \compratio \\
  \text{s.t.}  \quad& \beta_i-\alpha_i\geq 0\\
& \alpha_{i+1}+\beta_i \geq 1\\
& \compratio -\frac{3}{2}\alpha_{i+1} - \sum\limits_{j=0}^{i} \frac{3}{2^{i+2-j}} \beta_j \geq 0 \quad (i\in \{0, \dots, d-1\})\\
& \compratio, \alpha_0, \dots, \alpha_d, \beta_0, \dots, \beta_{d-1} \geq 0.
\end{align*}

Next, we show a lower bound on the objective value for (LP) (and, hence, a lower bound for the competitive ratio of $\alg$) by giving a solution for its dual, i.e., using weak duality. For this, we first formulate the dual of (LP). We denote the dual variables for the first constraint by $a_0, \dots, a_{d-1}$, for the second constraint by $b_0, \dots, b_{d-1}$, and for the last constraint by $c_0, \dots, c_{d-1}$, using the convention $a_d := 0$. Then the dual is given by

\begin{align*}
\mathrm{(D)}\qquad  \max \quad& \sum\limits_{i=0}^{d-1} b_i\\
\text{s.t.}  \quad& -a_0\leq 0\\
& - a_{i}+b_{i-1} - \frac{3}{2}c_{i-1} \leq 0 \quad (i\in \{1, \dots, d\})\\
& a_i + b_i  - \sum\limits_{j=i}^{d-1} \frac{3}{2^{j+2-i}} c_j \leq 0 \quad (i\in \{0, \dots, d-1\})\\
& \sum\limits_{j=0}^{d-1} c_j \leq 1\\
& a_i, b_i, c_i \geq 0 \quad (i\in \{0, \dots, d-1\}).
\end{align*}

Set $a_0= \dots = a_{d-1}=0, b_0=\dots=b_{d-1}=\frac{3}{2(d+1)}, c_0=\dots=c_{d-2}=\frac{1}{d+1}, c_{d-1}=\frac{2}{d+1}$.
We now show that this is a feasible solution. For the third inequality, we have for $i \in \{0, \dots, d-2\}$
\begin{align*}
a_i + b_i  - \sum\limits_{j=i}^{d-1} \frac{3}{2^{j+2-i}} c_j
&= \frac{3}{2(d+1)} -\sum\limits_{j=i}^{d-2} \frac{3}{2^{j+2-i}} \frac{1}{d+1} -\frac{3}{2^{d+1-i}}\frac{2}{d+1}\\
&= \frac{3}{2(d+1)} \left( 1- \sum\limits_{j=i}^{d-2} \frac{1}{2^{j+1-i}} - \frac{2}{2^{d-i}} \right)\\
&= \frac{3}{2(d+1)} \left( 1- \sum\limits_{j=1}^{d-1-i} \frac{1}{2^{j}} - \frac{1}{2^{d-1-i}} \right) =0.
\end{align*}
For $i=d-1$, we have in the third inequality
\begin{equation*}
a_i + b_i  - \sum\limits_{j=i}^{d-1} \frac{3}{2^{j+2-i}} c_j= \frac{3}{2(d+1)} - \frac{3}{2^2} \frac{2}{d+1}=0.
\end{equation*}
Hence, the third inequality is fulfilled. For all other inequalities, it is immediate that they are fulfilled.

Therefore, the described solution is feasible and its objective value is
\begin{equation*}
\sum\limits_{i=0}^{d-1} b_i = \frac{3d}{2(d+1)}.
\end{equation*}

Therefore, we obtain that the competitive ratio of $\alg$ on a binary tree of depth $d$ is at least $\frac{3d}{2(d+1)}$. Letting $d\to\infty$ yields the statement of the Theorem.
\end{proof}

\begin{remark}
It is worth noting that our lower bound construction also applies to the \emph{vertex arrival model}, where an ordering of the vertices $V=\{v_1, \dots, v_n\}$ is given and, in time step~$t$, all edges between $v_t$ and $v_1, \ldots, v_{t-1}$ are revealed.
Furthermore, our construction also applies to trees (not only forests), meaning that, at every time step, the graph induced by the revealed edges is connected.
\end{remark}

\bibliography{paper_actual}

\end{document}